\documentclass[11pt]{article}
\usepackage[a4paper, margin=1.2in]{geometry}
\usepackage[english]{babel}
\usepackage[T1]{fontenc}  			% make sure to install cm-super package for improved PDF quality
\usepackage{lmodern}
\usepackage[utf8]{inputenc}
\usepackage{microtype}			 	% microtypographical fine-tunning
\usepackage{graphicx}
\usepackage{setspace} 				% Provides support for setting the spacing between lines in a document.
\usepackage{tocbibind}				% Add bibliography/index/contents to Table of Contents.

\usepackage{bm} 					% bold math symbols including greek letters
\usepackage{booktabs}				% professional tables (w/o vertical bars)
\usepackage{array}					%
\usepackage{dcolumn} 				% align to decimal point in tables
\usepackage[shortcuts]{extdash} 	% hyphenation of dashed words

\usepackage[dvipsnames]{xcolor}
\usepackage{pgfplots}
\usepackage{authblk}        % For convenient typesetting of author affiliations
\usepackage{verbatim}       % For multiline commenting

\usepackage{amsfonts}
\usepackage{amsmath}
\usepackage{amsthm}
\usepackage{amssymb}
\usepackage{mathtools}
\usepackage{commath}
\usepackage{mathrsfs}

\usepackage[lined,vlined]{algorithm2e}
\usepackage[caption=false,font=footnotesize]{subfig}
\usepackage{soul}
\usepackage[sort,comma]{natbib}

\usepackage{hyperref}
\hypersetup{colorlinks=true, linkcolor=blue, citecolor=blue}
\usepackage{cleveref}
\usepackage{siunitx}

\newtheorem{theorem}{Theorem}

\newtheorem{proposition}{Proposition}

\newtheorem{remark}{Remark}

\newtheorem{assumption}{Assumption}

\begin{document}

\title{The two filter formula reconsidered: \\Smoothing in partially observed Gauss--Markov models without information parametrization}

\author{Filip Tronarp}

\date{Centre for Mathematical Sciences, Lund University \\
    \today}

\maketitle

\begin{abstract}
In this article, the two filter formula is re-examined in the setting of partially observed Gauss--Markov models.
It is traditionally formulated as a filter running backward in time,
where the Gaussian density is parametrized in ``information form''.
However, the quantity in the backward recursion is strictly speaking not a distribution, but a likelihood.
Taking this observation seriously, a recursion over log-quadratic likelihoods is formulated instead,
which obviates the need for ``information'' parametrization.
In particular, it greatly simplifies the square-root formulation of the algorithm.
Furthermore, formulae are given for producing the forward Markov representation of the
a posteriori distribution over paths from the proposed likelihood representation.
\end{abstract}

\section{Introduction}
Consider the following partially observed Markov process
\begin{subequations}\label{eq:pomp}
\begin{align}
x_0 &\sim \pi_0(\cdotp) \\
x_t \mid x_{t-1} &\sim \pi_{t \mid t-1}(\cdotp \mid x_{t-1}) \\
y_t \mid x_t &\sim h_{t \mid t}(\cdotp \mid x_t),
\end{align}
\end{subequations}
where $x_t$ is the state at time $t$, taking values in $\mathbb{R}^n$,
and $y_t$ is the measurement or observation at time $t$, taking values in $\mathbb{R}^m$, with $m \leq n$.
Let $t_i \leq \bar{t}_i$ for $i = 1, 2, 3$, denote their corresponding ranges by $t_i\colon\bar{t}_i$,
and denote the distribution of $x_{t_2:\bar{t}_2}$ conditioned on $x_{t_1:\bar{t}_1}$ and $y_{t_3:\bar{t}_3}$ by
\begin{equation}\label{eq:conditonal-state-distributions}
\pi^{t_3:\bar{t}_3}_{t_2:\bar{t}_2 \mid t_1 : \bar{t}_1}( x_{t_2: \bar{t}_2} \mid x_{t_1: \bar{t}_1}, y_{t_3: \bar{t}_3}),
\end{equation}
where the ranges $t_1:\bar{t}_1$ and $t_2:\bar{t}_2$ are non-overlapping.
Similarly, denote the distribution of $y_{t_2:\bar{t}_2}$ conditioned on $x_{t_1:\bar{t}_1}$ and $y_{t_3:\bar{t}_3}$ by
\begin{equation}\label{eq:conditonal-measurement-distributions}
h^{t_3:\bar{t}_3}_{t_2:\bar{t}_2 \mid t_1 : \bar{t}_1}(y_{t_2: \bar{t}_2} \mid x_{t_1: \bar{t}_1}, y_{t_3: \bar{t}_3})
\end{equation}
where the ranges $t_1:\bar{t}_1$ and $t_2:\bar{t}_2$ are non-overlapping.
Lastly, the distribution of $y_{t_2:\bar{t}_2}$ conditioned on $y_{t_3:\bar{t}_3}$ is denoted by
\begin{equation}\label{eq:conditional-marginal-measurement-distributions}
L_{t_2:\bar{t}_2 \mid t_3:\bar{t}_3}(y_{t_2: \bar{t}_2} \mid y_{t_3: \bar{t}_3}),
\end{equation}
where the ranges $t_1:\bar{t}_1$ and $t_2:\bar{t}_2$ are non-overlapping.
Given a sequence of measurements $y_{1:T}$ the marginal likelihood, $L_{1:T}$, and path posterior, $\pi^{1:T}_{0:T}$, are given by marginalization and Bayes' rule:
\begin{subequations}\label{eq:bayes-rule}
\begin{align}
L_{1:T}(y_{1:T}) &= \int \pi_0(x_0) \prod_{t=1}^T h_{t\mid t}(y_t \mid x_t) \pi_{t \mid t-1}(x_t \mid x_{t-1}) \dif x_{0:T} \label{eq:marginal-likelihood}\\
\pi^{1:T}_{0:T}(x_{0:T} \mid y_{1:T} ) &= L_{1:T}^{-1}(y_{1:T}) \pi_0(x_0) \prod_{t=1}^T h_{t \mid t}(y_t \mid x_t) \pi_{t \mid t-1}(x_t \mid x_{t-1}) \label{eq:path-posterior}.
\end{align}
\end{subequations}
As the measurement sequence shall be considered fixed it will henceforth be omitted in the expressions \eqref{eq:conditonal-state-distributions},
\eqref{eq:conditonal-measurement-distributions}, \eqref{eq:conditional-marginal-measurement-distributions}, and \eqref{eq:bayes-rule}.
The problem is thus to efficiently compute $L_{1:T}$ and interesting quantities from $\pi^{1:T}_{0:T}$, such as the smoothing distributions, $\pi^{1:T}_t$,
or the a posteriori transition distributions $\pi^{1:T}_{t \mid s}$.
There are essentially two methods to do this, the \emph{forward-backward} method \citep{Askar1981} and \emph{backward-forward} method \citep{Cappe2005},
a catalogue of different forward-backward and backward-forward methods is provided by \citet{AitElFquih2008}.
The forward-backward method has the advantage that it recursively computes the \emph{filtering distributions} $\pi^{1:t}_t$ forward in time;
it is thus suitable for on-line inference \citep{Kailath2000,Anderson2012,Sarkka2023}.
It can also compute the \emph{backward} transition distributions, $\pi^{1:t}_{t \mid t+1}$,
from which the smoothing distributions can be computed by the Rauch--Tung--Striebel algorithm \citep{Rauch1965}.
On the other hand, the backward-forward method computes the \emph{likelihoods of the future} $h_{t:T \mid t-1}$ backward in time.
It can also compute the \emph{forward} transition distributions, $\pi^{t:T}_{t \mid t-1}$,
from which the smoothing distributions can be computed once $\pi^{1:T}_0$ has been obtained (see theorem \ref{thm:backward-forward} below).
Both methods can also be combined to compute the smoothing distributions via the \emph{two-filter formula} \cite{Bresler1986}
\begin{equation}\label{eq:two-filter}
\pi^{1:T}_t(x) \propto h_{t+1:T\mid t}(x) \pi^{1:t}_t(x) \propto h_{t:T \mid t}(x) \pi^{1:t-1}_t(x)
\end{equation}
The two filter formula has been used to construct algorithms for approximate inference in Markovian switching systems \citep{Helmick1995}, jump Markov linear systems \citep{Balenzuela2022}, navigation problems \citep{Liu2010},
and in sequential Monte Carlo \citep{Sarkka2012,Lindsten2013,Lindsten2015}.
The abstract formulae for computing the likelihoods $h_{t:T\mid t}$, the a posteriori transition distributions $\pi^{t:T}_{t \mid t-1}$,
and the marginal likelihood, $L_{1:T}$ are given by theorem \ref{thm:backward-forward} \citep[chapter 3]{Cappe2005}.

\begin{theorem}[The backward-forward method]\label{thm:backward-forward}
Consider the partially observed Markov process in \eqref{eq:pomp},
then $h_{t:T \mid t}$ satisfy the following recursion:
\begin{subequations}\label{eq:likelihood-recursion}
\begin{align}
h_{t:T \mid t-1}(x_{t-1})   &= \int h_{t:T \mid t}(x_t) \pi_{t \mid t-1}(x_t \mid x_{t-1}) \dif x_t\\
h_{t-1:T \mid t-1}(x_{t-1}) &= h_{t-1\mid t-1}(x_{t-1}) h_{t:T \mid t-1}(x_{t-1}),
\end{align}
\end{subequations}
and the marginal likelihood, $L_{1:T}$ is given by
\begin{equation}
L_{1:T} = \int h_{1:T \mid 0}(x_0) \pi_0(x_0) \dif x_0.
\end{equation}
Additionally, the posterior over paths is given by
\begin{subequations}
\begin{align}
\pi^{1:T}_0(x_0) &= \frac{h_{1:T \mid 0}(x_0) \pi_0(x_0)}{L_{1:T}} \\
\pi^{t:T}_{t \mid t-1}(x_t \mid x_{t-1}) &= \frac{h_{t:T \mid t}(x_t) \pi_{t \mid t-1}(x_t \mid x_{t-1}) }{h_{t:T \mid t-1}(x_{t-1})} \\
\pi^{1:T}_{0:T}(x_{0:T}) &= \pi^{1:T}_0(x_0) \Pi_{t=1}^T  \pi^{t:T}_{t \mid t-1}(x_t \mid x_{t-1}).
\end{align}
\end{subequations}
\end{theorem}
In this article, the backward-forward algorithm is re-examined when \eqref{eq:pomp} is Gauss--Markov.
More specifically, $\pi_0$, $\pi_{t \mid t-1}$, and $h_{t \mid t}$ are then given by
\begin{subequations}\label{eq:gauss-markov}
\begin{align}
\pi_0(x) &= \mathcal{N}(x; \mu_0, \Sigma_0) \\
\pi_{t \mid t-1}(x \mid z) &= \mathcal{N}(x; \mu_{t \mid t-1}(z), Q_{t, t-1}) \\
h_{t \mid t}(y \mid x) &= \mathcal{N}(y; C_t x, R_t),
\end{align}
\end{subequations}
where the conditional mean function, $\mu_{t \mid t-1}$ is given by
\begin{equation}\label{eq:prior-transition-mean}
\mu_{t\mid t-1}(z) = \Phi_{t, t-1} z + u_{t, t-1}.
\end{equation}
The general idea has been around since \citet{Mayne1966,Fraser1969} and invovles treating $h_{t:T \mid t}$ as a distributon and computing it recursively backward in time
starting from $h_{T\mid T}$ by what is essentially a Kalman filter \citep{Kalman1961}.
The method is related to the Lagrange multiplier method for the maximum a posteriori trajectory \citep{Watanabe1985}.
However, $h_{t:T \mid t}$ is not a distribution, and may not be normalizable on $\mathbb{R}^n$ \citep{Cappe2005}.
This problem is circumvented by propagating an information vector, $\xi^{t:T}_t$, and information matrix, $\Lambda^{t:T}_t$ instead of a mean vector and covariance matrix \citep{Fraser1969}.
The information matrix is related to the covariance matrix by inversion.
More specifically, $h_{t:T \mid t}$ is parametrized according to \citep{Balenzuela2022}
\begin{equation}\label{eq:information-parametrization}
\begin{split}
\log h_{t:T \mid s}(x) &= \log \mathcal{N}\big(x; (\Lambda^{t:T}_{t \mid s})^{-1}\xi^{t:T}_{t \mid s},  (\Lambda^{t:T}_{t \mid s})^{-1}\big) + \texttt{const} \\
&= -\frac{1}{2}\langle x, \Lambda^{t:T}_{t \mid s} x \rangle + \langle \xi^{t:T}_{t \mid s} , x \rangle + \texttt{const}
\end{split}
\end{equation}
where the inverse is strictly formal since $\Lambda^{t:T}_{t \mid s}$ may not be invertible.
As remarked by \citet{Balenzuela2018,Balenzuela2022}, the constant of proportionality, say $\tilde{c}_{t \mid s}$, is often omitted from the computation.
This is troublesome in parameter estimation problems as it is required for computing the marginal likelihood, $L_{1:T}$ (c.f. theorem \ref{thm:backward-forward}).
The backward recursion for the parametrization \eqref{eq:information-parametrization} is essentially a so-called \emph{information filter} is run backwards in time \citep[sec 6.3]{Anderson2012}.
Some formulas in the literature require $\Phi_{t, t-1}$ to be invertible \citep[section 10.4.2]{Kailath2000}.
Other formulas require that $Q_{t, t-1}$ be invertible \citep[section 3.2]{Kitagawa2023}.
These assumptions are restrictive and not even required \citep{Fraser1969,Sarkka2020,Balenzuela2018,Balenzuela2022}.

\subsection{Contribution}
By regarding the backward method as a recursion over likelihoods rather than as a recursion over distributions,
an algorithm is obtained which obviates the need for ``information'' parametrization.
More specifically, $h_{t:T \mid t}$ is parametrized according to
\begin{equation*}
h_{t:T \mid t}(x) = c_{t:T\mid t} \exp \Big[ -\frac{1}{2} \abs[0]{\bar{y}_{t:T\mid t} - \bar{C}_{t:T\mid t} x}^2 \Big].
\end{equation*}
This has the consequence that $h_{t:T\mid t-1}$ and $\pi^{t:T}_{t\mid t-1}$ can be computed by a method that is algebraically equivalent to the classical Kalman update in mean/covariance parametrization.
In particular, the author is not aware of anyone giving a method for computing $\pi^{t:T}_{t\mid t-1}$ from the backwards recursion.
The price to be paid for this convenience is that the formula for computing $h_{t-1:T\mid t-1}$ from $h_{t:T\mid t-1}$ and $h_{t \mid t}$ is slightly more complicated.
More specifically, a QR factorization is required to ensure that the dimensionality $\bar{m}_{t\mid t}$ of $\bar{y}_{t:T\mid t}$ does not grow beyond $n$.
Additionally, a method only operating on covariance square-roots is developed, assuming that the square-roots $R_t^{1/2}$ and $Q_{t, t-1}^{1/2}$ are given.
\footnote{However, $R^{1/2}$ is already required for the first method.}
\footnote{When the estimation problem originates from a continuous-discrete problem, then $Q_{t, t-1}^{1/2}$ may be obtained directly by a recent algorithm \citep{Stillfjord2023}.}

The remainder of this article is organizes as follows.
In section \ref{sec:new-recursion} the forward-backward method in the proposed likelihood parametrization is developed,
in section \ref{sec:cholesky} this is generalized to only operate on Cholesky factors of covariances,
and in section \ref{sec:two-filter} it is shown that the two filter formula becomes algebraically equivalent to the classical Kalman update in the proposed parametrization.
Section \ref{sec:likelihood-as-density} discusses the interpretation of the likelihood as a density and maximum likelihood estimation of the state based on future measurements.
In section \ref{sec:experiment}, the proposed algorithm is demonstrated in a state estimation problem with flat initial distribution,
and a concluding discussion is given in section \ref{sec:discussion}.

\section{The new backward-forward recursion}\label{sec:new-recursion}
In this section, a method for implementing theorem \ref{thm:backward-forward} is derived.
The novelity lies in that $h_{t:T \mid s}$ for $s = t-1, t$ is identified with a linear Gaussian likelihood with identity covariance.
Consequently, the likelihood recursion \eqref{eq:likelihood-recursion} reduces to a recursion over a measurement vector $\bar{y}_{t:T\mid s}$ and a measurement matrix $\bar{C}_{t:T \mid s}$.
The following assumption is not essential to carry out this program, but simplifies the exposition greatly.

\begin{assumption}\label{ass:psd-measurement-covariance}
The measurement covariance matrices $\{R_t\}_{t=1}^T$ are positive definite.
\end{assumption}

\subsection{The backward recursion}
Let $y$ be a vector and $C$ and $R$ matrices of appropriate dimensions.
Denote by $\mathcal{Q}(x; y, C, R)$ the following quadratic function:
\begin{equation}\label{eq:quadratic-form}
\mathcal{Q}(x; y, C, R) = (y - C x)^* R^{-1} (y - C x).
\end{equation}
When $R$ is the identity matrix, $R = \mathrm{I}$, it is omitted from the notation like so:
\begin{equation}
\mathcal{Q}(x; y, C, \mathrm{I}) = \mathcal{Q}(x; y, C).
\end{equation}
Note that when $R$ is positive definite then \eqref{eq:quadratic-form} is overparametrized in the sense that
$R$ has a unique Cholesky factorization $R = R^{1/2} R^{*/2}$, therefore
\footnote{
Note that $\mathcal{Q}(x; y, C)$ is still overparametrized for general $C$.
By the QR factorization, $C = Q_C R_C$ so $\mathcal{Q}(x; y, C) = \mathcal{Q}(x; Q_C^* y, R_C)$ since $Q_C$ is unitary.
This is to some extent taken care of by proposition \ref{prop:big-likelihood-update}.
}
\begin{equation}\label{eq:white-quadratic-form}
\begin{split}
\mathcal{Q}(x; y, C, R) &=  (y - C x)^* R^{-1} (y - C x) = \abs[0]{ R^{-1/2}y - R^{-1/2} C x}^2 \\
&= \mathcal{Q}(x; R^{-1/2}y, R^{-1/2} C).
\end{split}
\end{equation}
This section is dedicated to proving the following theorem and in the process deriving a recursion for the parameters in \eqref{eq:posterior-transition-expression} and \eqref{eq:likelihood-expression}.
The method of proof is by induction.
\begin{theorem}\label{thm:gauss-markov-likelihood}
Let assumption \ref{ass:psd-measurement-covariance} be in effect,
then $\pi^{t:T}_{t\mid t-1}$ and $h_{t:T \mid s}$ for  $s = t-1, t$ and $t \leq T$ have the following representation:
\begin{subequations}
\begin{align}
\mu^{t:T}_{t \mid t-1}(z)         &=  \Phi^{t:T}_{t, t-1} z + u^{t:T}_{t, t-1} \\
\pi^{t:T}_{t \mid t-1}(x \mid z ) &= \mathcal{N}(x; \mu^{t:T}_{t \mid t-1}(z) , Q^{t:T}_{t, t-1}) \label{eq:posterior-transition-expression} \\
h_{t:T\mid s}(x) &= c_{t:T\mid s} \exp \Big[ -\frac{1}{2} \mathcal{Q}(x; \bar{y}_{t:T\mid s}, \bar{C}_{t:T \mid s}) \Big] \label{eq:likelihood-expression},
\end{align}
\end{subequations}
where  $\bar{y}_{t:T\mid s} \in \mathbb{R}^{\bar{m}_{t:T\mid s}}$ and $\bar{C}_{t:T\mid s} \in \mathbb{R}^{\bar{m}_{t:T\mid s} \times n}$ and $m \leq \bar{m}_{t:T\mid s} \leq n$.
\end{theorem}
The formula \eqref{eq:likelihood-expression} clearly hold when $T = t = s$, since by  assumption \ref{ass:psd-measurement-covariance} and \eqref{eq:white-quadratic-form}
\begin{equation*}
\begin{split}
h_{t \mid t}(x) &=  \mathcal{N}(y_t; C_t x, R_t) = \abs{2\pi R_t}^{-1/2} \exp \Big[ - \frac{1}{2} \mathcal{Q}(x; y_t, C_t, R_t) \Big] \\
&= \abs{2\pi R_t}^{-1/2} \exp \Big[ - \frac{1}{2} \mathcal{Q}(x; R_t^{-1/2} y_t, R_t^{-1/2} C_t) \Big]
 = c_{t\mid t} \exp \Big[ - \frac{1}{2} \mathcal{Q}(x; \bar{y}_{t\mid t}, \bar{C}_{t\mid t}) \Big],
\end{split}
\end{equation*}
where the parameters are given by
\begin{subequations}\label{eq:white-likelihoods}
\begin{align}
c_{t\mid t} &= \abs{2\pi R_t}^{-1/2}, \\
\bar{y}_{t\mid t} &= R_t^{-1/2} y_t, \\
\bar{C}_{t\mid t} &= R_t^{-1/2} C_t.
\end{align}
\end{subequations}
In particular that means that the terminal condition to the the recurrence in theorem \ref{thm:backward-forward}, $h_{T:T\mid T} = h_{T \mid T}$, satisfies \eqref{eq:likelihood-expression}.
The following proposition shows that \eqref{eq:posterior-transition-expression} and \eqref{eq:likelihood-expression} hold for $s = t-1$ provided that \eqref{eq:likelihood-expression} hold for $s = t$.
\begin{proposition}\label{prop:likelihood-prediction}
Assume that the likelihood, $h_{t:T \mid t}$, satisfies \eqref{eq:likelihood-expression},
then the likelihood, $h_{t:T \mid t-1}$, and the transition distribution, $\pi^{t:T}_{t \mid t-1}$,
satisfies \eqref{eq:likelihood-expression} and \eqref{eq:posterior-transition-expression}, respectively, with parameters:
\begin{subequations}\label{eq:likelihood-prediction}
\begin{align}
\hat{R}_{t:T \mid t-1} &= \mathrm{I} + \bar{C}_{t:T \mid t} Q_{t, t-1} \bar{C}_{t:T \mid t}^*, \\
\bar{y}_{t:T\mid t-1}  &=\hat{R}_{t:T \mid t-1}^{-1/2} \big( \bar{y}_{t:T\mid t} - \bar{C}_{t:T \mid t} u_{t, t-1} \big)  \\
\bar{C}_{t:T \mid t-1} &=  \hat{R}_{t:T \mid t-1}^{-1/2}\bar{C}_{t:T \mid t} \Phi_{t, t-1} \\
\log c_{t:T \mid t-1}  &= \log c_{t:T \mid t-1}  - \frac{1}{2} \log \abs[0]{\hat{R}_{t:T\mid t-1}} \\
\bar{K}^{t:T}_{t \mid t-1} &= Q_{t, t-1} \bar{C}_{t:T \mid t}^* \hat{R}_{t:T\mid t-1}^{-1} \\
\Phi^{t:T}_{t, t-1} &= \big( \mathrm{I} - \bar{K}^{t:T}_{t \mid t-1} \bar{C}_{t:T \mid t} \big) \Phi_{t, t-1} \\
u^{t:T}_{t, t-1}    &= u_{t, t-1} + \bar{K}^{t:T}_{t \mid t-1} \big( \bar{y}_{t:T\mid t} - \bar{C}_{t:T \mid t} u_{t, t-1}\big)  \\
Q^{t:T}_{t, t-1}    &= Q_{t, t-1} - \bar{K}^{t:T}_{t \mid t-1} \hat{R}_{t:T \mid t-1} (\bar{K}^{t:T}_{t \mid t-1})^*,
\end{align}
\end{subequations}
where $\bar{y}_{t:T\mid t-1} \in \mathbb{R}^{\bar{m}_{t:T\mid t-1}}$, $\bar{C}_{t:T\mid s} \in \mathbb{R}^{\bar{m}_{t:T\mid t-1} \times n}$, and  $\hat{R}_{t:T\mid s} \in \mathbb{R}^{\bar{m}_{t:T\mid t-1} \times n}$
with
\begin{equation}
\bar{m}_{t:T\mid t-1} = \bar{m}_{t:T\mid t}.
\end{equation}
\end{proposition}

\begin{comment}
\begin{remark}
The minimal eigenvalue of $\hat{R}_{t:T \mid t-1}$ is bounded from below by 1.
Therefore, computing the Cholesky factor, $\hat{R}_{t:T \mid t-1}^{1/2}$, is guaranteed to succeed in unless $n$ is very large (discounting the possibiltiy of underflow) \citep{Demmel1989}.
Furthermore, solving linear systems with respect to $\hat{R}_{t:T \mid t-1}$ is numerically stable \citep{Higham1996}.
The only numerically difficult part is thus the computation of $Q^{t:T}_{t, t-1}$.
This difficulty will be allieviated in section \ref{sec:cholesky} by means of square-root factorizations of covariance matrices.
\end{remark}
\end{comment}

\begin{proof}
The likelihood, $h_{t:T \mid t}$, is by assumption given by
\begin{equation*}
h_{t:T \mid t}(x) = c_{t:T \mid t} \exp\Big[ -\frac{1}{2} \mathcal{Q}(x; \bar{y}_{t:T\mid t}, \bar{C}_{t:T\mid t}) \Big] =   c_{t:T \mid t} (2\pi)^{\bar{m}_{t \mid t}} \mathcal{N}(\bar{y}_{t:T\mid t}; \bar{C}_{t:T\mid t} x, \mathrm{I}).
\end{equation*}
The computation of $h_{t:T \mid t-1}$ and $\pi^{t:T}_{t\mid t-1}$ is then by by theorem \ref{thm:backward-forward}, up to a multiplicative constant, equivalent to inference and marginalization in the following model
\begin{align*}
x_t \mid x_{t-1} & \sim \mathcal{N}( \mu_{t \mid t-1}(x_{t-1}), Q_{t, t-1}) \\
\bar{y}_{t:T \mid t} \mid x_t, x_{t-1} & \sim \mathcal{N}(\bar{y}_{t:T\mid t}; \bar{C}_{t:T\mid t} x_t, \mathrm{I}),
\end{align*}
conditionally on $x_{t-1}$.
The solution is given by \citep[c.f. lemma A.2 and A.3]{Sarkka2023}
\begin{align*}
x_t \mid x_{t-1}, \bar{y}_{t:T \mid t} &\sim \mathcal{N}( \mu^{t:T}_{t\mid t-1}(x_{t-1}) , Q^{t:T}_{t, t-1}) \\
\bar{y}_{t:T \mid t} \mid x_{t-1} & \sim \mathcal{N}\big( \bar{C}_{t:T\mid t} \big(\Phi_{t, t-1} x_{t-1} + u_{t, t-1} \big), \hat{R}_{t:T \mid t-1}\big),
\end{align*}
with the conditional mean function given by
\begin{equation*}
\begin{split}
\mu^{t:T}_{t\mid t-1}(x) &= \mu_{t\mid t-1}(x) + \bar{K}^{t:T}_{t \mid t-1}\big( \bar{y}_{t:T\mid t} - \bar{C}_{t:T \mid t} \mu_{t, t-1}(x)\big) \\
&= \Phi^{t:T}_{t, t-1} x + u^{t:T}_{t, t-1}.
\end{split}
\end{equation*}
The last equality is obtained by using \eqref{eq:prior-transition-mean} and simply matching terms.
Furthermore, re-introducing the constant of proportionality gives $h_{t:T\mid t-1}$ as
\begin{equation*}
\begin{split}
h_{t:T \mid t-1}(x) &=   c_{t:T \mid t} (2\pi)^{\bar{m}_{t \mid t}} \mathcal{N}\big(\bar{y}_{t:T \mid t}; \bar{C}_{t:T\mid t} \big(\Phi_{t, t-1} x + u_{t, t-1} \big), \hat{R}_{t:T \mid t-1}\big) \\
&=  c_{t:T \mid t} (2\pi)^{\bar{m}_{t \mid t}} \mathcal{N}(\bar{y}_{t:T \mid t-1}; \bar{C}_{t:T\mid t-1} x , \mathrm{I}) \\
&= c_{t:T \mid t} \abs[0]{\hat{R}_{t:T\mid t-1}}^{-1/2} \exp \Big[ -\frac{1}{2} \mathcal{Q}(x; \bar{y}_{t:T\mid t-1}, \bar{C}_{t:T\mid t-1}) \Big],
\end{split}
\end{equation*}
which concludes the proof.
\end{proof}
With proposition \ref{prop:likelihood-prediction} established, it remains to show that if $h_{t:T \mid t-1}$ satisfies \eqref{eq:likelihood-expression},
then so does $h_{t-1:T\mid t-1}$. Under this presupposition, then by theorem \ref{thm:backward-forward} and \eqref{eq:white-likelihoods},
$h_{t-1:T\mid t-1}$ is given by
\begin{equation}\label{eq:likelihood-update-expression}
\begin{split}
h_{t-1:T\mid t-1}(x) &= h_{t:T\mid t-1}(x) h_{t-1 \mid t-1}(x) \\
&=  c_{t:T\mid t-1}  \exp \Big[ -\frac{1}{2} \mathcal{Q}(x; \bar{y}_{t:T\mid t-1}, \bar{C}_{t:T \mid t-1}) \Big] \\
&\quad \times c_{t - 1\mid t - 1} \exp \Big[ -\frac{1}{2} \mathcal{Q}(x; \bar{y}_{t-1 \mid t-1}, \bar{C}_{t-1 \mid t-1}) \Big]
\end{split}
\end{equation}
Now let $\hat{m}_{t-1} = \bar{m}_{t:T\mid t-1} + m$ and define $\hat{y}_{t:T\mid t-1} \in \mathbb{R}^{\hat{m}_{t-1}}$ and $\hat{C}_{t:T\mid t-1} \in \mathbb{R}^{\hat{m}_{t-1} \times n}$ by
\begin{subequations}\label{eq:maybe-likelihood-update}
\begin{align}
\hat{y}_{t-1:T\mid t-1} &= \begin{pmatrix} \bar{y}_{t:T \mid t-1} \\ \bar{y}_{t-1 \mid t-1} \end{pmatrix} \\
\hat{C}_{t-1:T\mid t-1} &= \begin{pmatrix} \bar{C}_{t:T \mid t-1} \\ \bar{C}_{t-1 \mid t-1} \end{pmatrix},
\end{align}
\end{subequations}
then the sum of quadratic functions in \eqref{eq:likelihood-update-expression} evalutes to
\begin{equation*}
\mathcal{Q}(x; \bar{y}_{t:T\mid t-1}, \bar{C}_{t:T \mid t-1})  + \mathcal{Q}(x; \bar{y}_{t-1 \mid t-1}, \bar{C}_{t-1 \mid t-1}) = \mathcal{Q}(x; \hat{y}_{t-1:T\mid t-1}, \hat{C}_{t-1:T\mid t-1}).
\end{equation*}
If $\hat{m}_{t-1} \leq n$ this immediately gives an expression for $h_{t-1:T\mid t-1}$ which satisfies the statement of theorem \ref{thm:backward-forward}.
The following proposition summarizes this observation.
\begin{proposition}\label{prop:small-likelihood-update}
Assume $h_{t:T \mid t-1}$ satisfies \eqref{eq:likelihood-expression} and that $\hat{m}_{t-1} = \bar{m}_{t:T\mid t-1} + m \leq n$,
then $h_{t-1:T \mid t-1}$ satisfies \eqref{eq:likelihood-expression} with parameters
\begin{align}
\bar{y}_{t-1:T\mid t-1} &= \hat{y}_{t-1:T\mid t-1} \\
\bar{C}_{t-1:T\mid t-1} &= \hat{C}_{t-1:T\mid t-1} \\
c_{t-1:T\mid t-1} &= c_{t:T\mid t-1} c_{t \mid t}, \\
m_{t-1:T\mid t-1} &= \bar{m}_{t:T\mid t-1} + m.
\end{align}
\end{proposition}

The overdetermined case, when $\bar{m}_{t:T\mid t-1} + m > n$, requires slightly more effort.

\begin{proposition}\label{prop:big-likelihood-update}
Assume $h_{t:T \mid t-1}$ satisfies \eqref{eq:likelihood-expression} and $\hat{m}_{t-1} =  \bar{m}_{t:T\mid t-1} + m > n$.
Let $V_t U_t$ be a QR factorization of $\hat{C}_{t-1:T\mid t-1}$ and partition the factors according to
\begin{equation}
\hat{C}_{t-1:T\mid t-1} = V_t U_t = \begin{pmatrix} V_{t,1} & V_{t,2} \end{pmatrix} \begin{pmatrix} \bar{C}_{t-1:T\mid t-1} \\ 0_{(\hat{m}_{t-1} - n) \times n} \end{pmatrix},
\end{equation}
where $V_{t, 1} \in \mathbb{R}^{\hat{m}_{t-1} \times n}$, $V_{t, 2} \in \mathbb{R}^{\hat{m}_{t-1} \times (\hat{m}_{t-1} - n)}$, and $\bar{C}_{t-1:T\mid t-1} \in \mathbb{R}^{n \times n}$.
Furthermore, define $\bar{y}_{t-1:T\mid t-1}$ and $e_{t-1}$ by
\begin{align}
\bar{y}_{t-1:T\mid t-1} &= V_{t, 1}^* \hat{y}_{t-1:T\mid t-1} \\
e_{t-1} &= V_{t, 2}^* \hat{y}_{t-1:T\mid t-1},
\end{align}
then $h_{t-1:T\mid t-1}$ satisfies \eqref{eq:likelihood-expression} with remaining parameters given by
\begin{align}
c_{t-1:T\mid t-1} &= c_{t:T\mid t-1} c_{t \mid t} \exp\Big[ -\frac{ \abs[0]{e_{t-1}}^2 }{2} \Big] \\
\bar{m}_{t-1:T\mid t-1} &= n.
\end{align}
\end{proposition}

\begin{proof}
By assumption and theorem \ref{thm:backward-forward}, $h_{t-1:T\mid t-1}$ is given by
\begin{equation*}
\begin{split}
h_{t-1:T \mid t-1}(x) &= h_{t:T \mid t-1}(x) h_{t-1\mid t-1}(x) \\
&= c_{t:T\mid t-1} c_{t \mid t} \exp \Big[ -\frac{1}{2} \mathcal{Q}(x; \hat{y}_{t-1:T\mid t-1}, \hat{C}_{t-1:T \mid t-1}) \Big] \\
&= c_{t:T\mid t-1} c_{t \mid t} \exp \Big[ -\frac{1}{2} \mathcal{Q}(x; \hat{y}_{t-1:T\mid t-1},  V_t U_t ) \Big] \\
&= c_{t:T\mid t-1} c_{t \mid t} \exp \Big[ -\frac{1}{2} \mathcal{Q}(x; V_t^* \hat{y}_{t-1:T\mid t-1},  U_t ) \Big] \\
&= c_{t:T\mid t-1} c_{t \mid t} \exp\Big[ -\frac{ \abs[0]{e_{t-1}}^2 }{2} \Big] \exp \Big[-\frac{1}{2} \mathcal{Q}(x; \bar{y}_{t-1:T\mid t-1},  \bar{C}_{t-1:T\mid t-1} ) \Big],
\end{split}
\end{equation*}
which is the desired result.
\end{proof}

With propositions \ref{prop:likelihood-prediction}, \ref{prop:small-likelihood-update}, and \eqref{prop:big-likelihood-update} established,
theorem \ref{thm:gauss-markov-likelihood} hold by induction.
It remains to give formulae for the forward recursion. This is done in the following.

\subsection{The forward recursion}
The previous section established expressions for $\pi^{t:T}_{t \mid t-1}$ and  $h_{t:T\mid s}$ for $s = t-1, t$,
the parameters of which are obtained by alternating between proposition \ref{prop:likelihood-prediction} and
proposition \ref{prop:small-likelihood-update} or \ref{prop:big-likelihood-update}.
It remains to compute the a posteriori time marginals, $\{\pi^{1:T}_t\}_{t=0}^T$, and  the marginal likelihood, $L_{1:T}$.
They are all clearly Gaussian and will be denoted by
\begin{equation}
\pi^{1:T}_t(x) = \mathcal{N}(x; \mu^{1:T}_t, \Sigma^{1:T}_t), \quad t = 0, 1, \ldots, T
\end{equation}
According to theorem \ref{thm:backward-forward}, $L_{1:T}$ and $\pi^{1:T}_0$ are simply obtained by Bayes' rule.
By the same argument as in proposition \ref{prop:likelihood-prediction}, \emph{mutatis mutandis}, they are given by
\begin{subequations}\label{eq:marginal-likelihood-and-posterior-initial-distribution}
\begin{align}
S_0 &= \bar{C}_{1:T \mid 0} \Sigma_0 \bar{C}_{1:T \mid 0}^* + \mathrm{I} \\
K^{1:T}_0 &= \Sigma_0 \bar{C}_{1:T \mid 0}^* S_0^{-1} \\
\mu^{1:T}_0 &= \mu_0 + K^{1:T}_0 \big( \bar{y}_{1:T\mid 0} - \bar{C}_{1:T\mid 0} \mu_0 \big) \\
\Sigma^{1:T}_0 &= \Sigma_0 - K^{1:T}_0S_0 (K^{1:T}_0)^* \\
L_{1:T} &= c_{1:T\mid 0}(2\pi)^{\bar{m}_{t \mid 0}/2}\mathcal{N}( \bar{y}_{1:T \mid 0}; \bar{C}_{1:T\mid 0} \mu_0, S_0)
\end{align}
\end{subequations}
Furthermore, the remaining time marginals are then obtained by
\begin{subequations}\label{eq:posterior-time-marginal-moments}
\begin{align}
\mu^{1:T}_t &= \Phi^{t:T}_{t,t-1} \mu^{1:T}_{t-1} + u^{t:T}_{t, t-1} \\
\Sigma^{1:T}_t &=  \Phi^{t:T}_{t,t-1} \Sigma^{1:T}_{t-1} (\Phi^{t:T}_{t,t-1})^* + Q^{t:T}_{t, t-1}.
\end{align}
\end{subequations}

\section{The backward-forward recursion in Cholesky form}\label{sec:cholesky}
In order to ensure excellent numerical performance of the algorithm it is preferable to only work with Cholesky factors of the covariance matrices \citep[section 12]{Kailath2000}.
Computing $h_{t-1:T\mid t-1}$ from $h_{t-1:T\mid t-1}$ and $h_{t-1\mid t-1}$ does not involve any covariances as at all.
Hence the only problem in the backward recursion is to compute $h_{t:T\mid t-1}$ from $h_{t:T \mid t}$ and $\pi_{t\mid t-1}$ without ever forming a covariance matrix.
This can be done by so-called array algorithms \citep{Kailath2000,Anderson2012},
the key idea is that for a given matrix $A$ an upper triangular Cholesky factor of $A^* A$ may be obtained by the QR decomposition $A = QU$.
For an arbitrary matrix $A$, write $A \cong U$ when $U$ is the upper triangular factor in the QR factorization of $A$.

\begin{proposition}\label{prop:likelihood-prediction-cholesky}
Consider the same setting as propositon \ref{prop:likelihood-prediction},
then
\begin{equation}\label{eq:array-algorithm}
\begin{pmatrix}
\mathrm{I} & 0 \\
Q_{t, t-1}^{*/2} \bar{C}_{t:T\mid t}^* & Q_{t, t-1}^{*/2}
\end{pmatrix}
\cong
\begin{pmatrix}
\hat{R}_{t:T\mid t-1}^{*/2} & (\hat{K}^{t:T}_{t\mid t-1})^* \\
0  & (Q_{t, t-1}^{t:T} )^{*/2}
\end{pmatrix},
\end{equation}
where $\hat{K}^{t:T}_{t\mid t-1}$ is related to $\bar{K}^{t:T}_{t\mid t-1}$ by
\begin{equation}\label{eq:whitened-gain}
\bar{K}^{t:T}_{t\mid t-1} = \hat{K}^{t:T}_{t\mid t-1} \hat{R}_{t:T\mid t-1}^{-1/2},
\end{equation}
and the parameters of $h_{t:T \mid t-1}$ are given by
\begin{align}
\bar{y}_{t:T\mid t-1} &= \hat{R}_{t:T\mid t-1}^{-1/2} \big( \bar{y}_{t:T\mid t} - \bar{C}_{t:T\mid t} u_{t, t-1} \big) \\
\bar{C}_{t:T \mid t-1} &= \hat{R}_{t:T\mid t-1}^{-1/2} \bar{C}_{t:T \mid t} \Phi_{t, t-1} \\
\log c_{t:T \mid t-1}  &= \log c_{t:T \mid t-1}  -  \log \abs[0]{\hat{R}_{t:T\mid t-1}^{1/2}}.
\end{align}
Furthermore, the remaining parameters of $\pi^{t:T}_{t\mid t-1}$ are given by
\begin{align}
\Phi^{t:T}_{t, t-1} &= \Phi_{t, t-1} - \hat{K}^{t:T}_{t\mid t-1} \bar{C}_{t:T \mid t-1} \\
u^{t:T}_{t, t-1}    &= u_{t, t-1} + \hat{K}^{t:T}_{t\mid t-1} \bar{y}_{t:T\mid t-1}.
\end{align}
\end{proposition}

\begin{remark}
Pay attention to the time index of $\bar{y}$ and $\bar{C}$ in the expressions for $\Phi^{t:T}_{t, t-1}$ and $u^{t:T}_{t, t-1}$ in proposition \ref{prop:likelihood-prediction-cholesky}.
It differs from the corresponding expression in proposition \ref{prop:likelihood-prediction}.
\end{remark}

\begin{proof}
The identities \eqref{eq:array-algorithm} and \eqref{eq:whitened-gain} are known \citep[section 6.5]{Anderson2012}.
The proof concludes by inserting these into the relations given by proposition \ref{prop:likelihood-prediction} and identifying terms.
\end{proof}

The parameters of the initial distribution are obtained by the same argument as proposition \ref{prop:likelihood-prediction-cholesky}.
More specifically, substitue $Q_{t, t-1}$ and $\bar{C}_{t:T\mid t}$ on the left-hand side of \eqref{eq:array-algorithm}
with $\Sigma_0$ and $\bar{C}_{1:T \mid 0}$ to obtain $S_0^{1/2}$,  $\hat{K}_0^{1:T}$, and $(\Sigma^{1:T}_0)^{*/2}$ as the right-hand side of \eqref{eq:array-algorithm}.
The remaining parameter, $\mu^{1:T}_0$, is then obtained by using the relation between $\hat{K}_0^{1:T}$ and $K_0^{1:T}$ corresponding to \eqref{eq:whitened-gain}
and substitution into \eqref{eq:marginal-likelihood-and-posterior-initial-distribution}
\begin{equation}
\mu^{1:T}_0 = \mu_0 + \hat{K}_0^{1:T} S_0^{-1/2} \big( \bar{y}_{1:T\mid 0} - \bar{C}_{1:T\mid 0} \mu_0 \big).
\end{equation}
Furthermore, the mean of the smoothing marginals is still obtained by \eqref{eq:posterior-time-marginal-moments},
while the Cholesky factor of the covariances clearly satisifes the following recursion:
\begin{equation}
\begin{pmatrix} \Phi^{t:T}_{t,t-1} (\Sigma^{1:T}_{t-1})^{1/2} &  (Q^{t:T}_{t, t-1})^{1/2} \end{pmatrix}^* \cong (\Sigma^{1:T}_t)^{*/2}.
\end{equation}
That is, $(\Sigma^{1:T}_t)^{*/2}$ is the upper triangular factor i  the QR factorization of the left-hand side of the preceeding expression.

\section{The two-filter formula}\label{sec:two-filter}
Assume the filtering densities $\pi^{1:t}_t$ have been obtained, by for example running a standard Kalman filter \citep{Sarkka2023} or a square-root variant \citep[section 6.5]{Anderson2012},
then they are given by
\begin{equation}
\pi^{1:t}_t(x) = \mathcal{N}(x; \mu_t^{1:t}, \Sigma^{1:t}_t).
\end{equation}
Furthermore, by theorem \ref{thm:gauss-markov-likelihood}, the two-filter formula \eqref{eq:two-filter} reduces to
\begin{equation}
\pi^{1:T}_t(x) \propto h_{t+1:T\mid t}(x) \pi^{1:t}_t(x)
\propto \mathcal{N}(\bar{y}_{t+1:T\mid t}; \bar{C}_{t+1:T\mid t}x, \mathrm{I}) \mathcal{N}(x; \mu^{1:t}_t, \Sigma^{1:t}_t).
\end{equation}
This is equivalent to obtaining the posterior in the following model,
\begin{align}
x_t &\sim \mathcal{N}(\mu^{1:t}_t, \Sigma^{1:t}_t), \\
\bar{y}_{t+1:T\mid t} \mid x_t &\sim \mathcal{N}(\bar{y}_{t+1:T\mid t}; \bar{C}_{t+1:T\mid t} x_t, \mathrm{I}),
\end{align}
hence the mean, $\mu^{1:T}_t$ and covariance, $\Sigma^{1:T}_t$, of $\pi^{1:t}_t$ can be obtained by a classical Kalman update or its' square-root variant. \citep{Anderson2012}.

\section{Relationship to information parametrization, maximum likelihood estimation, and the likelihood as a density}\label{sec:likelihood-as-density}
The purpose of this section is to relate the present parametrization of $h_{t:T \mid s}$ with the information parametrization and to clarify its interpretation as a density.
From the preceeding discussion, the likelihood $h_{t:T\mid s}$ for $s = t-1, t$ is given by
\begin{equation}
h_{t:T\mid s}(x) = c_{1:T \mid s} \exp \Big[ -\frac{1}{2} \mathcal{Q}(x; \bar{y}_{t:T\mid s}, \bar{C}_{t:T\mid s}) \Big].
\end{equation}
The quadratic form above is given by
\begin{equation}
\mathcal{Q}(x; \bar{y}_{t:T\mid s}, \bar{C}_{t:T\mid s}) = \langle x, \bar{C}_{t:T\mid s}^* \bar{C}_{t:T\mid s} x \rangle - 2\langle \bar{C}_{t:T\mid s}^*  \bar{y}_{t:T\mid s}, x \rangle + \mathtt{const}.
\end{equation}
The correspondence to the information parametrization is then (c.f. \eqref{eq:information-parametrization})
\begin{subequations}
\begin{align}
\xi^{t:T}_s &= \bar{C}_{t:T\mid s}^*  \bar{y}_{t:T\mid s}, \\
\Lambda^{t:T}_s &= \bar{C}_{t:T\mid s}^* \bar{C}_{t:T\mid s}.
\end{align}
\end{subequations}
Hence, $\bar{C}_{t:T\mid s}$ may be viewed as a $\bar{m}_{t\mid s} \times n$ square-root of $\Lambda^{t:T}_s$.
The relationship between $\xi^{t:T}_s$ and $\bar{y}_{t:T\mid s}$ is a bit more opaque and perhaps best understood in terms of the mean/covariance parametrization.
Let  $\mu^{t:T}_s$ be the minimum norm solution to the following equation
\begin{equation}
0 = \nabla \log h_{t:T\mid s}(x)  =  \bar{C}_{t:T\mid s}^* (\bar{y}_{t:T\mid s} - \bar{C}_{t:T\mid s} x),
\end{equation}
then $\mu^{t:T}_s$ is the maximum likelihood estimate of $x_s$ based on data $y_{t:T}$ of smallest norm.
\footnote{
Since the rank of $\bar{C}_{t:T\mid s}$ may be smaller than $n$, the maximum likelihood estimate is in general not unique.
}
Now define $\Sigma^{t:T}_s$ by
\begin{equation}
\Sigma^{t:T}_s = (\Lambda^{t:T}_s)^+ = ( \bar{C}_{t:T\mid s}^* \bar{C}_{t:T\mid s})^+ = \bar{C}_{t:T\mid s}^+ (\bar{C}_{t:T\mid s}^*)^+
\end{equation}
where $^+$ denotes a generalized matrix inverse and the last equality is taken known from e.g. \citet{Greville1966}.
An explicit expression of $\mu^{t:T}_s$ is then given by
\begin{equation}\label{eq:whitened-mean}
\mu^{t:T}_s = \Sigma^{t:T}_s \bar{C}_{t:T\mid s}^* \bar{y}_{t:T\mid s}
=  \bar{C}_{t:T\mid s}^+ (\bar{C}_{t:T\mid s}^*)^+ \bar{C}_{t:T\mid s}^* \bar{y}_{t:T\mid s}
= (\Sigma^{t:T}_s)^{1/2} (\bar{C}_{t:T\mid s}^*)^+ \bar{C}_{t:T\mid s}^* \bar{y}_{t:T\mid s}.
\end{equation}
The observation vector $\bar{y}_{t:T\mid s}$ has identity covariance conditioned on $x_s$ (by construction),
hence the conditional covariance of $\mu^{t:T}_s$ is given by
\begin{equation}
\Sigma^{t:T}_s \bar{C}_{t:T\mid s}^*   \bar{C}_{t:T\mid s} \Sigma^{t:T}_s = (\Lambda^{t:T}_s)^+ \Lambda^{t:T}_s (\Lambda^{t:T}_s)^+ =  (\Lambda^{t:T}_s)^+ =  \Sigma^{t:T}_s,
\end{equation}
which justifies the notation $\Sigma^{t:T}_s$ for $(\Lambda^{t:T}_s)^+$.
In view of \eqref{eq:whitened-mean}, the observation vector $\bar{y}_{t:T\mid s}$ may then be interpreted as a the mean vector $\mu^{t:T}_s$ ``whitened'' by $\Sigma^{t:T}_s$.

It is now tempting to interpret $h_{t:T\mid s}$ as an unnormalized Gaussian density with mean $\mu^{t:T}_s$ and covariance $\Sigma^{t:T}_s$.
This can indeed be done, but the density will not not be supported on the whole of $\mathbb{R}^n$ but rather on an affine subset, $\Omega_{t:T \mid s}$,
which is given by \citep[Chapter 8]{Rao1973}
\begin{equation}
\Omega_{t:T \mid s} = \big\{ x \colon x = \mu^{t:T}_s + v, \quad v \in \operatorname{range} \Sigma^{t:T}_s \big\}.
\end{equation}
The integral of $h_{t:T\mid s}$ over $\Omega_{t:T \mid s}$ is then given by
\begin{equation}
\begin{split}
\int_{\Omega_{t:T \mid s}} h_{t:T\mid s}(x) \dif x &=
c_{1:T \mid s} \int_{\Omega_{t \mid s}} \exp \Big[ -\frac{1}{2} (x - \mu^{t:T}_s)^* (\Sigma^{t:T}_s)^+ (x - \mu^{t:T}_s) \Big]\dif x  \\
&= c_{1:T \mid s} \abs[0]{2\pi \Sigma^{t:T}_s}^{1/2}_+,
\end{split}
\end{equation}
where $\abs{\cdotp}_+$ denotes the product of all non-zero eigenvalues, the so-called pseudo-determinant.
The notation $\dif x$ is here overloaded to also mean the Lebesgue measure on $\Omega_{t:T \mid s}$.
Thus the inference problem can be solved even when $\pi_0$ is flat (improper) on $\Omega_{t \mid s}$.
In particular, the following result is obtained by the preceeding discussion.

\begin{proposition}\label{prop:improper-initial-distribution}
Let the a priori initial distribution be given by
\begin{equation}
\pi_0(x) = \chi_{\Omega_{1:T \mid 0}}(x),
\end{equation}
where $\chi_A$ is the indicator function on the set $A$.
Then the marginal likelihood, $L_{1:T}$, and the a posteriori initial density, $\pi^{1:T}_0$, are given by
\begin{subequations}
\begin{align}
L_{1:T}  &=  c_{1:T \mid 0} \abs[0]{2\pi \Sigma^{t:T}_0}^{1/2}_+ \\
\pi^{1:T}_0(x) &=
\begin{cases}
\mathcal{N}(x; \mu^{t:T}_0, \Sigma^{t:T}_0), \quad x \in \Omega_{1:T \mid 0} \\
0,  \quad x \in \Omega_{1:T \mid 0}^{\mathsf{c}}
\end{cases}.
\end{align}
\end{subequations}
\end{proposition}

Note that $\pi_0^{1:T}$ can still be obtained by similar means as in proposition \ref{prop:improper-initial-distribution}
when $\pi_0$ is flat on all of $\mathbb{R}^n$.
The difference being that $\pi^{1:T}_0$ will remain flat on $\Omega_{1:T \mid 0}^{\mathsf{c}}$ rather than being unsupported.
However, the marginal likelihood, $L_{1:T}$, will be infinite.
Therefore, in such a siutation it is perhaps more appropriate to treat $x_0$ as a parameter in which case the marginal likelihood may
identified with $h_{1:T\mid 0}$.

\section{An example of unknown origin}\label{sec:experiment}
Consider an object moving in the plane and denote its position at time $\tau$ by $p(\tau)$.
Let the dynamics of the object be governed by the following continuous-time model
\begin{equation}\label{eq:object-dynamics}
\dif \ddot{p}(\tau) = \begin{pmatrix} \sigma_1 & 0 \\ 0 & \sigma_2 \end{pmatrix} \dif w(\tau), \quad \tau \geq 0,
\end{equation}
where $\ddot{p}$ is the acceleration and $w$ is a two dimensional standard Wiener process.
Starting at time $\tau = 127$, the position of the object is measured at every time unit up until $\tau = 256$ according to
\begin{equation}\label{eq:object-measurements}
y(\tau_k) = p(\tau_k) + \begin{pmatrix} \sqrt{\lambda_1} & 0 \\ 0 & \sqrt{\lambda_2} \end{pmatrix} v(\tau_k),
\end{equation}
where $v(\tau_k)$ is a sequence of two dimensional independent standard Gaussian vectors, $\tau_k = k$.
Collecting the acceleration, $\ddot{p}(\tau)$, the velocity, $\dot{p}(\tau)$, and the position, $p(\tau)$,
into a state vector $x(\tau)$ gives a continuous-discrete stochastic state-space model representation of \eqref{eq:object-dynamics} and \eqref{eq:object-measurements}.
Furthermore, by discretizing the continuous-time process and setting $x_k = x(\tau_k)$ and taking $x_0$ as completely unknown
gives the following discrete-time model:
\begin{subequations}
\begin{align}
x_0 &\sim \chi_{\mathbb{R}^6}(\cdotp) \\
x_k \mid x_{k-1} &\sim \mathcal{N}( \Phi x_{k-1}, Q), \quad k = 1, 2, \ldots, 256, \\
y_k \mid x_k &\sim \mathcal{N}(C x_k, R), \quad k = 127, 128, \ldots, 256.
\end{align}
\end{subequations}
The problem is to determine the position, $p(\tau_k)$, of the object for $k=0,1,\ldots, 126$,
that is at times before the measurement process started.

Two options based on the backward-forward method are considered:
\begin{enumerate}
\item The complete forward-backward algorithm using proposition \ref{prop:improper-initial-distribution} to compute $\pi_0^{127:256}$
\item Computing the maximum likelihood estiamte of $x_k$ by maximizing $h_{127:256 \mid k}$.
\end{enumerate}

The results are shown in figures \ref{fig:posterior} and \ref{fig:mle}, respectively.
As expected, the estimation quality of running the complete forward-backward algorithm is better than just doing maximum likelihood estimation.

\begin{figure}[t!]
\includegraphics{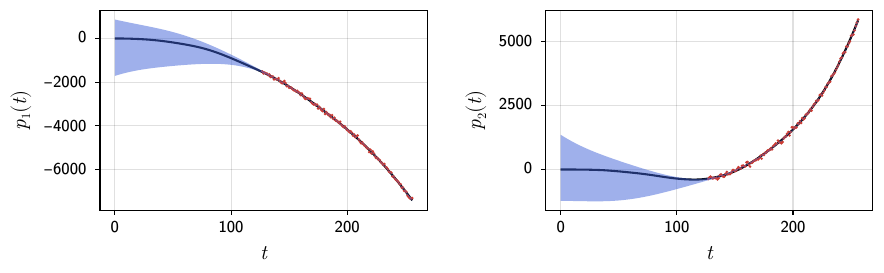}
\caption{
The position of the object (black), the position observations (red), and a $\pm 2\sigma$ credible interval of the position (blue).
The estimate was obtained by the forward-backward algorithm.
}\label{fig:posterior}
\end{figure}

\begin{figure}[t!]
\includegraphics{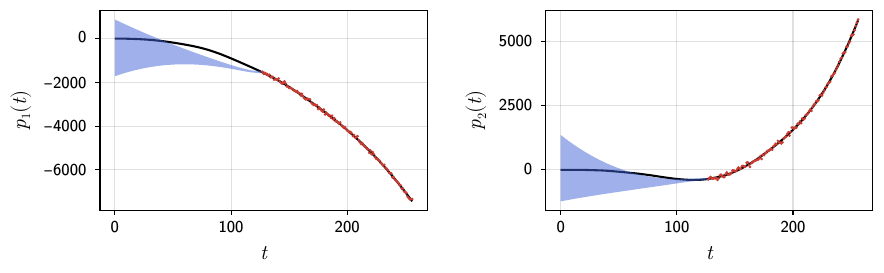}
\caption{
The position of the object (black), the position observations (red), and a $\pm 2\sigma$ confidence interval of the position (blue).
The estimate was obtained by maximum likelihood.
}\label{fig:mle}
\end{figure}

\section{Discussion}\label{sec:discussion}
This article presents a new backward-forward method was developed for inference in partially observed Gauss--Markov models.
The benefit is that all computations are carried out with respect to the covariance parametrization or its' square-root variant.
Furthermore, simple expressions for the \emph{forward} a posteriori transition distributions, $\pi^{t:T}_{t \mid t-1}$, were obtained.
This in turn gives a forward Markov representation of the path posterior, $\pi^{1:T}_{0:T}$,
which may be simpler to work with in certain applications such as parameter estimation using variational methods.

\section*{Acknowledgements}
FT was partially supported by the Wallenberg AI, Autonomous Systems and Software Program (WASP) funded by the Knut and Alice Wallenberg Foundation.

\appendix

\bibliographystyle{apalike}
\bibliography{refs}

\end{document}